\documentclass[conference]{IEEEtran}

\makeatletter
\def\ps@headings{%
\def\@oddhead{\mbox{}\scriptsize\rightmark \hfil \thepage}%
\def\@evenhead{\scriptsize\thepage \hfil \leftmark\mbox{}}%
\def\@oddfoot{}%
\def\@evenfoot{}}
\makeatother
\usepackage{amsfonts}
\usepackage{mathrsfs}
\usepackage{amsfonts}
\usepackage{amssymb}
\usepackage{graphicx}
\usepackage{epsfig}
\usepackage{psfrag}
\usepackage{amsmath}
\usepackage{array}
\usepackage{cases}
\usepackage{subfigure}
\usepackage{cite,graphicx,amsmath,amssymb,color}
\usepackage{anysize}
\usepackage{algorithmic}
\usepackage{algorithm}

\usepackage{algorithmic}
\usepackage{stmaryrd}
\usepackage{multirow}
\usepackage{subfig}
\usepackage{graphicx,times,amsmath} 

\usepackage{url}

\marginsize{13mm}{13mm}{19mm}{43mm}

\IEEEoverridecommandlockouts

\begin{document}

\title{Optimization-Based Linear Network Coding for General Connections of Continuous Flows}

\author{\authorblockN{Ying Cui\thanks{Y. Cui was supported in part by the  National Science Foundation of China grant 61401272.  M. M\'{e}dard was supported by the Air Force Office of Scientific Research (AFOSR) grant  FA9550-13-1-0023. E. Yeh was supported by the National Science Foundation grant CNS-1423250.}
}\authorblockA{Dept. of EE}\authorblockA{Shanghai Jiao Tong University}\and\authorblockN{Muriel M\'{e}dard}\authorblockA{Dept. of EECS}
\authorblockA{MIT}\and\authorblockN{Edmund Yeh}\authorblockA{Dept. of ECE} \authorblockA{Northeastern University}
\and \authorblockN{Douglas Leith}\authorblockA{Dept. of CS}
\authorblockA{Trinity College Dublin}\and \authorblockN{Ken Duffy}\authorblockA{Hamilton Institute}
\authorblockA{NUI Maynooth}
}

\maketitle

\newtheorem{Thm}{Theorem}
\newtheorem{Lem}{Lemma}
\newtheorem{Cor}{Corollary}
\newtheorem{Def}{Definition}
\newtheorem{Exam}{Example}
\newtheorem{Alg}{Algorithm}
\newtheorem{Sch}{Scheme}
\newtheorem{Prob}{Problem}
\newtheorem{Rem}{Remark}
\newtheorem{Proof}{Proof}
\newtheorem{Asump}{Assumption}
\newtheorem{Subp}{Subproblem}

\begin{abstract} For general connections, the problem  of finding network codes and optimizing resources for those codes  is intrinsically difficult and little is known about its complexity.   Most of the existing  solutions  rely on very restricted classes of network codes in terms of the number of flows allowed to be coded together, and are not entirely distributed. 
In this paper, we consider a new method for constructing linear network codes for general connections of continuous flows to minimize the total cost of  edge use based on mixing.  
We first formulate the minimum-cost network coding design problem. To solve the optimization problem, we propose two equivalent alternative formulations with discrete mixing and continuous mixing, respectively, and develop distributed algorithms to solve them. 
Our approach allows fairly general coding across flows and guarantees no greater cost than any solution without  inter-flow network coding.
\end{abstract}

\section{Introduction}

In the case of general connections (where each destination can request information from any subset of sources), 
the problem of finding network codes  is intrinsically difficult.
Little is known about its complexity and  its decidability remains unknown. 
In certain special cases, such as multicast connections (where destinations share all of their demands), it is sufficient to satisfy a Ford-Fulkerson type of min-cut max-flow constraint between all sources to every destination individually. 
For multicast connections, linear codes are sufficient \cite{KM03, LYC03}  and  a distributed random construction exists\cite{Hoetal06}. 
In the literature, linear codes have been the most widely considered. However,  in general, linear codes over finite fields may not be sufficient for general connections \cite{DFZ05}.  
In addition, even when we consider simple scalar network codes (with scalar coding coefficients),  the problem of code construction  for general connections (i.e., neither multicast nor its variations)  remains vexing\cite{twounicastKamath2014}.  
The main difficulty lies in canceling the effect of flows that are
coded together but not destined for a common
destination.

The problem of code construction becomes more involved
when we seek to limit the use of network links for reasons
of network resource management. 
In the case of multicast connections of  continuous flows, it is known that finding a minimum-cost  solution for convex cost functions of  flows over edges of the network is  a convex optimization problem  and can be solved distributively using convex decomposition\cite{Lunetal06}. 
In the case of
general connections of continuous flows, however,
network resource minimization, even when considering
only restricted code constructions, appears to be difficult.


In general, there are two types of coding approaches for optimizing network use for general connections. The first type of coding is  mixing, which consists of coding together flows from sources using the random linear distributed code construction of \cite{Hoetal06} (originally proposed for multicast connections),  as though the flows were parts of a common multicast connection.
In this case, no explicit code coefficients are provided and decodability is ensured with high probability by the random coding, given that mixing is properly designed.  For example, in \cite{Lun04networkcoding}, a two-step mixing approach is proposed for network resource minimization of general connections, where flow partition and flow rate optimization are considered separately.  In \cite{CUI2015ISITreport}, we introduce linear network mixing coefficients and present a new method for constructing linear network codes for general connections of integer flows to minimize the total cost of edge use. The minimum-cost network coding design problem in \cite{CUI2015ISITreport} is a discrete optimization problem, which jointly considers mixing and flow optimization.  
The second type of coding is an explicit
linear code construction, where one provides specific
linear coefficients, to be applied to flows at different nodes, over
some finite field. 
In this case, the explicit linear code constructions are usually  simplified by restricting them to be binary, generally in the context of coding flows together only pairwise.  For example, in \cite{POEM10} and \cite{ShroffBButterfly}, simple two-flow combinations  are proposed for network resource minimization of  general connections. 

The flow rate optimization in \cite{Lun04networkcoding},  the joint mixing and flow optimization in \cite{CUI2015ISITreport},  and the joint two-flow coding and flow optimization in \cite{POEM10, ShroffBButterfly}  can be solved distributively. However, the separation of flow partition and flow rate optimization
in \cite{Lun04networkcoding} and the pairwise coding  in \cite{POEM10, ShroffBButterfly}  lead in general to   feasibility region  reduction and  network cost increase.  In \cite{CUI2015ISITreport}, we do not allow flow splitting and coding over time, leading to coded symbols flowing through each edge of the network at an integer rate. The restriction of integer flow rates affects the network cost reduction.

In this paper, we consider a new method for constructing linear network codes to minimize the total cost of edge use for satisfying general connections of continuous flows.  We generalize the linear network  mixing coefficients introduced in \cite{CUI2015ISITreport}.  In contrast to\cite{CUI2015ISITreport}, we allow flow splitting and coding over time to further reduce network cost. 
Using mixing  with generalized mixing coefficients,  we  formulate the minimum-cost network coding design problem, which is an instance of mixed discrete-continuous programming. Our mixing-based formulation allows for fairly general coding across flows, offers a tradeoff between performance and computational complexity via tuning a design parameter controlling the mixing effect, and guarantees no greater cost than any solution without  inter-flow network coding. To solve the mixed discrete-continuous optimization problem, we propose two equivalent alternative formulations with discrete mixing and continuous mixing, respectively, and develop distributed algorithms to solve them. Specifically, the distributed algorithm for the discrete mixing formulation  is obtained by relating  the optimization problem to a constraint satisfaction problem (CSP) in discrete optimization and applying recent results in the domain \cite{cfl}. The distributed algorithm for the continuous mixing formulation is based on penalty methods for nonlinear programming\cite{Bertsekasbooknonlinear:99} in continuous optimization.

\section{Network Model and Definitions}
In this section, we first illustrate the network model  for general connections of continuous flows.  The model is similar to  that we considered in \cite{CUI2015ISITreport} for integer flows, except that here we consider general flow rates and edge capacities, and allow flow splitting and coding over time. 
Next, to facilitate the understanding of the formulations proposed in Sections \ref{sec:cont-mix}, \ref{subsec:alt-cont-mix-cont}, and \ref{subsec:alt-cont-mix-dis}, we  also briefly illustrate  the formal relationship between linear network coding and mixing established in \cite{CUI2015ISITreport}.


\subsection{Network Model}

We consider a directed acyclic network with general connections. Let $\mathcal G=(\mathcal V, \mathcal E)$ denote the  directed acyclic  graph, where $\cal V$ denotes the set of $V=|\cal V|$ nodes and $\cal E$ denotes the set of $E=|\cal E|$ edges.  To simplify notation, we assume there is only one edge from  node $i\in\mathcal V$ to node $j\in \mathcal V$, denoted as edge $(i,j)\in\mathcal E$.\footnote{Multiple edges from node $i$ to node $j$ can be modeled by introducing multiple extra nodes, one on each edge.}  
For each node $i\in \cal V$, define the set of incoming neighbors to be $\mathcal I_i=\{j: (j,i)\in \mathcal E \}$ and the set of outgoing neighbors to be $\mathcal O_i=\{j:(i,j)\in \mathcal E\}$.  Let $I_i=|\mathcal I_i|$ and $O_i=|\mathcal O_i|$ denote the in degree and out degree of node $i\in \mathcal V$, respectively. Assume $I_i \leq D$ and $O_i\leq D$ for all $i \in \mathcal V$. 
Let $\mathcal P=\{1,\cdots, P\}$ denote the set of $P=|\mathcal P|$ flows to be carried by the network.  For each flow $p\in \mathcal P$, let $s_p\in \mathcal V$ be its source.  We consider continuous flows. Let  $R_p\in \mathbb R^+$ denote the source rate for source $p$,  where $\mathbb R^+$ denotes the set of non-negative real numbers. 
Let $\mathcal S=\{s_1,\cdots, s_P\}$ denote the set of $P=|\mathcal S|$ sources. 
To simplify notation, we assume different flows do not share a common source node and no source node has any incoming edges. 
Let $\mathcal T=\{t_1,\cdots, t_T\}$ denote the set of $T=|\mathcal T|$ terminals. Each terminal $t\in \mathcal T$ demands a subset of $P_t=|\mathcal P_t|$ flows $\mathcal P_t\subseteq \mathcal P$. Assume each flow is requested by at least one terminal, i.e., $\cup_{t\in\mathcal T}\mathcal P_t=\mathcal P$.  To simplify notation, we assume no terminal has any outgoing edges.


 Let  $B_{ij}\in \mathbb R^+$ denote the edge capacity for edge $(i,j)$.  
 Let $z_{ij}\in [0,B_{ij}]$ denote the transmission rate through  edge $(i,j)$. 
 We assume a cost is incurred on an edge when information is transmitted through the edge. Let $U_{ij}(z_{ij})$ denote the cost function incurred on edge $(i,j)$ when the transmission rate through  edge $(i,j)$ is $z_{ij}$. Assume $U_{ij}(z_{ij})$ is  convex, non-decreasing, and twice continuously differentiable in $z_{ij}$.  
 We are interested in the problem of finding  linear network coding designs and  minimizing the network cost $\sum_{(i,j)\in \mathcal E} U_{ij}(z_{ij})$ of general connections of continuous flows for those designs.

\subsection{Scalar Time-Invariant Linear Network Coding and Mixing}\label{subsec:coding-mixing}

For ease of exposition, in this section, we   illustrate linear network coding and mixing  by considering   unit flow rate, unit edge capacity and one (coded) symbol transmission for each edge per unit time, and adopt scalar time-invariant notation. Later, in Sections \ref{sec:cont-mix}, \ref{subsec:alt-cont-mix-cont}, and \ref{subsec:alt-cont-mix-dis}, 
we shall consider general flow rates and edge capacities  and allow flow splitting and coding over time, which enable multiple (coded) symbols to flow through each edge at a continuous rate. 

Consider a  finite field $\mathcal F$ with size $F=|\mathcal F|$.  
 In linear network coding,    a linear combination  over $\mathcal F$ of the symbols in $\{\sigma_{ki}\in \mathcal F:k\in \mathcal I_i\}$ from the incoming edges $\{(k,i):k\in \mathcal I_i\}$, i.e., 
$\sigma_{ij}=\sum_{k\in \mathcal I_i}\alpha_{kij}\sigma_{ki}
$,   can be transmitted through the shared edge $(i,j)$, where coefficient $\alpha_{kij}\in \mathcal F$ is referred to  as the local coding coefficient corresponding to edge $(k,i)\in \mathcal E$ and edge $(i,j)\in \mathcal E$. 
On the other hand, the symbol of  edge  $(i,j)\in \mathcal E$ can be expressed as a linear combination 
over $\mathcal F$ of the source symbols $\{\sigma_p\in \mathcal F:p\in \mathcal P\}$, i.e., $\sigma_{ij}=\sum_{p\in \mathcal P}c_{ij,p}\sigma_p$, where coefficient $c_{ij,p}\in \mathcal F$ is referred to as the global coding coefficient of  flow $p\in \mathcal P$ and edge $(i,j)\in \mathcal E$. Let $\mathbf c_{ij}=(c_{ij,1},\cdots, c_{ij,p},\cdots,c_{ij,P})\in \mathcal F^P$ denote the $P$ coefficients corresponding to this linear combination for edge $(i,j)\in \mathcal E$, referred to as the global coding vector of edge $(i,j)\in \mathcal E$. 
Note that, we consider scalar  time-invariant linear network coding. 
In other words, $\alpha_{kij}\in \mathcal F$ and $c_{ij,p}\in \mathcal F$ are both scalars, and $\alpha_{kij}$ and $c_{ij,p}$ do not change over time.    
When using scalar linear network coding, for each terminal,  extraneous  flows are allowed to be mixed with the  desired flows on the paths to the terminal, as the  extraneous flows can be cancelled at intermediate nodes or at the terminal.


In many cases, we shall see that the specific values of the local or global coding coefficients   are not required in  our design.  For this purpose, we introduce the  mixing concept based on local and global mixing coefficients  established in \cite{CUI2015ISITreport}. Specifically, we consider  the local mixing  coefficient $\beta_{kij}\in \{0,1\}$ corresponding to edge $(k,i)\in \mathcal E$ and edge $(i,j)\in \mathcal E$,
which relates to the local coding coefficient 
$\alpha_{kij}\in \mathcal F$ as follows. 
$\beta_{kij}=1$ indicates that symbol $\sigma_{ki}$ of edge $(k,i)\in \mathcal E$ is allowed to contribute to the linear combination over $\mathcal F$ forming symbol $\sigma_{ij}$  and $\beta_{kij}=0$ otherwise. Thus, if $\beta_{kij}=0$, we have  $\alpha_{kij}=0$  (note that $\alpha_{kij}$ can be zero when $\beta_{kij}=1$).
Similarly, we consider the global mixing coefficient $x_{ij,p}\in\{0,1\}$ of  flow $p\in \mathcal P$ and edge $(i,j)\in \mathcal E$, which relates to the global coding coefficient $c_{ij,p}\in\mathcal F$ as follows. $x_{ij,p}=1$ indicates that flow $p$ is allowed to be  mixed (coded) with other flows, i.e.,  symbol $\sigma_p$ is allowed to contribute to the linear combination over $\mathcal F$ forming symbol $\sigma_{ij}$, and $x_{ij,p}=0$ otherwise. Thus, if $x_{ij,p}=0$, we have $c_{ij,p}=0$ (note that $c_{ij,p}$ can be zero when $x_{ij,p}=1$). Then, we introduce the global mixing vector $\mathbf x_{ij}=(x_{ij,1},\cdots, x_{ij,p},\cdots, x_{ij,P})\in \{0,1\}^P$  for edge $(i,j)\in \mathcal E$, which relates to the global coding vector $\mathbf c_{ij}=(c_{ij,1},\cdots, c_{ij,p},\cdots,c_{ij,P})\in \mathcal F^P$. 
Similarly, we consider scalar time-invariant linear network mixing.  That is, $\beta_{kij}\in \{0,1\}$ and $x_{ij,p}\in \{0,1\}$ are both scalars, and $\beta_{kij}$ and $x_{ij,p}$ do not change over time.   
 

\begin{figure}[h]
\begin{center}
\includegraphics[height=2.3cm]{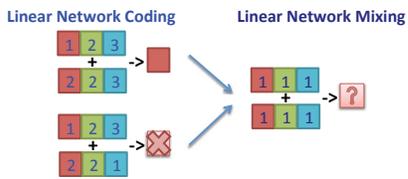}
\caption{Comparisons between linear network coding and linear network mixing. 
Given the two global coding vectors on the left, we can tell  that the (red) symbol can be decoded in the first case and cannot be decoded in the second case. However, given the two global mixing vectors on the right (same for the two cases on the left),   it is not sufficient to tell whether the  (red) symbol  can be decoded or not.}\label{Fig:coding-mixing}
\end{center}
\end{figure}

Global mixing vectors provide a natural way of speaking of flows as possibly coded or not coded without knowledge of the specific values of the global coding vectors. Intuitively,  global mixing vectors can be regarded as a limited  representation  of global coding vectors.  Network mixing vectors may not be sufficient for telling  whether a certain symbol can be decoded or not,  as illustrated in Fig.~\ref{Fig:coding-mixing}. 
Therefore, using the network mixing representation, extraneous flows which are  mixed with the desired flows on the paths to each terminal, are not guaranteed to be cancelled at the terminal. Let $\mathbf e_p$ denote the vector with the $p$-th element being 1 and all the other elements being 0. Let $\vee$ denote the  ``or'' operator (logical disjunction). We now define
the feasibility for scalar linear network mixing.
\begin{Def} [Feasibility of Scalar Linear Network Mixing] For a network $\mathcal G=(\mathcal V, \mathcal E)$ and a set of flows $\mathcal P$ with sources $\mathcal S$  and  terminals $\mathcal T$, a linear network mixing design $\{\mathbf x_{ij}\in \{0,1\}^P: (i, j)\in \mathcal E\}$ is called feasible if the following three conditions are satisfied: 
1) $\mathbf x_{s_pj}=\mathbf e_p$ for source edge $(s_p,j)\in \mathcal E$, where $s_p\in \mathcal S$ and $ p\in \mathcal P$; 
2) $\mathbf x_{ij}=\vee_{k\in \mathcal I_i} \beta_{kij}\mathbf x_{ki}$ for edge  $(i,j) \in \mathcal E$ not outgoing from a source, where $i\not\in \mathcal S$ and $\beta_{kij}\in \{0,1\}$; 
3) $\vee_{i\in \mathcal I_t} x_{it,p}=1$ for all  $p \in \mathcal P_t$ and  $x_{it,p}=0$ for all  $i\in \mathcal I_t$ and $p \not\in\mathcal P_t$, where $t\in \mathcal T$.
\label{Def:feasibility-mixing}
\end{Def}

Note that Condition 3) in Definition \ref{Def:feasibility-mixing} ensures that for each terminal,  the extraneous flows are not  mixed with the desired flows on the paths to the terminal. In other words, using   linear network mixing, only mixing is allowed at intermediate nodes. This is not as general as using linear network coding, which allows both mixing and canceling  (i.e., removing one or multiple flows from a
mixing of flows) at intermediate nodes.

Given a feasible  linear network mixing design (specified by $\{\beta_{kij}\in\{0,1\}:(k,i),(i,j)\in\mathcal E\}$), one way to accomplish mixing when $\mathcal F$ is large is to use random linear network coding \cite{Hoetal06} (to obtain $\{\alpha_{kij}\in\mathcal F:(k,i),(i,j)\in\mathcal E\}$), as discussed in the introduction.  Note that, in performing random linear network coding based on $\beta_{kij}$, $\alpha_{kij}$ can be randomly chosen in $\mathcal F$ when $\beta_{kij}=1$, but $\alpha_{kij}$ must be chosen to be 0  when $\beta_{kij}=0$.

\section{Continuous Flows  with Mixing Only}\label{sec:cont-mix}

In this section, we consider the minimum-cost scalar time-invariant  linear network coding design problem   for general connections  of continuous flows with mixing only.   Starting from this section,  we consider multiple global mixing vectors (each may correspond  to multiple global coding vectors) for each edge and allow coded symbols to flow through each edge at a continuous rate.

%
%
%

\subsection{Design Parameter}


We consider multiple global mixing vectors for each edge.  This generalizes the linear network  mixing coefficients introduced in \cite{CUI2015ISITreport}. 
We refer to the number of global network mixing vectors  for each edge as the mixing parameter, denoted as $L\in\{1,\cdots, L_{\max}\}$, where $L_{\max}$ is the maximum number of global network mixing vectors necessary for  decodability   using mixing (cf. Definition \ref{Def:feasibility-mixing}),  and is given as follows.
Let $\boldsymbol{\mathcal Y}$ denote the set of atoms of the algebra generated by $\{\mathcal P_t: t\in \mathcal T\}$, i.e., $\boldsymbol{\mathcal Y}\triangleq\{\cap_{t\in \mathcal T} \mathcal Y_t: \mathcal Y_t=\mathcal P_t \ \text{or}\ \mathcal Y_t=\mathcal P-\mathcal P_t\}-\{\emptyset\}.$ In other words, $\boldsymbol{\mathcal Y}$ gives a set partition of $\mathcal P$ that represents the flows that can be mixed (cf. Definition \ref{Def:feasibility-mixing}) over an edge in the worst case (i.e., all terminals obtaining flows through the same edge). We choose 
$L_{\max}=|\boldsymbol{\mathcal Y}|$. 
Note that $1\leq L_{\max}\leq P$, where $L_{\max}=1$ for the multicast case, i.e.,  $\mathcal P_t=\mathcal P$ for all $t\in \mathcal T$, and $L_{\max}=P$ for the unicast case, i.e., $\mathcal P_{t'}\cap \mathcal P_t=\emptyset$ for all $t\neq t'$ and $t,t'\in \mathcal T$. 
Fig. \ref{Fig:para_cont} illustrates an example of flow partition $\boldsymbol{\mathcal Y}$ and mixing parameter $L$ for the general case.  

\begin{figure}[h]
\begin{center}
\includegraphics[height=3cm, width=3cm]{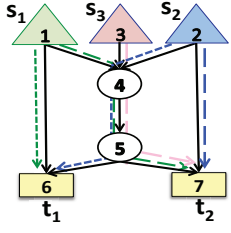}
\caption{\small{Illustration of a feasible solution to Problem \ref{Prob:new-low-cont-multi}.   
$\mathcal P=\{1,2,3\}$, $\mathcal S=\{s_1,s_2,s_3\}$,  $R_1=R_2=R_3=1$, $B_{ij}=2$ for all  $(i,j)\in \mathcal E$, $\mathcal T=\{t_1,t_2\}$, $\mathcal P_1=\{1,2\}$ and $\mathcal P_2=\{1,2,3\}$. Thus,   $\boldsymbol{\mathcal Y}=\{\{1,2\},\{3\}\}$, $L_{\max}=|\boldsymbol{\mathcal Y}|=2$ and $L\in\{1,2\}$.  
 }}\label{Fig:para_cont}
\end{center}
\end{figure}


For a given mixing parameter $L$, we now introduce the global and local network mixing vectors. 
For each $l=1,\cdots,L$, let $\mathbf x_{ij}(l)=(x_{ij,1}(l),\cdots, x_{ij,p}(l),\cdots, x_{ij,P}(l))\in \{0,1\}^P$  denote the $l$-th global network mixing vector over edge $(i,j)\in \mathcal E$. 
Let $\boldsymbol{\beta}_{kij}(l,m)\in \{0,1\}$ denote the local mixing  coefficient corresponding to the $l$-th global network mixing vector of  edge $(k,i)\in \mathcal E$ (i.e., $\mathbf x_{ki}(l)$) and  the $m$-th global network mixing vector of  edge $(i,j)\in \mathcal E$ (i.e., $\mathbf x_{ij}(m)$),  where $l,m=1,\cdots,L$.

\subsection{Problem Formulation}

We would like to find the minimum-cost scalar time-invariant   linear network coding   design with design parameter $L\in\{1,\cdots, L_{\max}\}$ for general connections  of continuous flows  with mixing only.

\begin{Prob} [Continuous Flows with Mixing Only] 
\begin{align}
 {U_{x}^*}(L)&=\min_{\substack{\{z_{ij}\},\{z_{ij}(l)\} \{f_{ij,p}^{t}(l)\}\\ \{  x_{ij,p}(l)\},\{\beta_{kij}(l,m)\}}}\quad  \sum_{(i,j)\in \mathcal E} U_{ij}(z_{ij})\nonumber\\
s.t.\ 
& 0\leq z_{ij}\leq B_{ij}, \ (i,j)\in \mathcal E\label{eqn:mix-z-multi}\\
& x_{ij,p} (l)\in \{0,1\}, \   l=1,\cdots, L,\ (i,j)\in \mathcal E,\ p\in \mathcal P\label{eqn:mix-x-multi}\\
& \beta_{kij}(l,m)\in \{0,1\}, \ l,m=1,\cdots, L, \ (k,i),  (i,j) \in \mathcal E \label{eqn:mix-beta-multi}\\
&f_{ij,p}^{t} (l)\geq 0, \ l=1,\cdots, L,\ (i,j)\in \mathcal E,\  p \in \mathcal P_t, \ t\in \mathcal T\label{eqn:mix-f-multi}\\
& \sum_{p\in\mathcal P_t}  f_{ij,p}^{t}(l) \leq z_{ij}(l),   \ t\in \mathcal T , \ l=1,\cdots,L, \ (i,j)\in \mathcal E\label{eqn:mix-f-z-multi}\\
& \sum_{l=1}^L z_{ij}(l)\leq z_{ij},  \   (i,j)\in \mathcal E\label{eqn:mix-f-z-sum-multi}\\
&\sum_{k\in \mathcal O_i}\sum_{l=1}^L f_{ik,p}^{t}(l) -\sum_{k\in \mathcal I_i} \sum_{l=1}^L f_{ki,p}^{t}(l)=\sigma_{i,p}^{t}, \nonumber\\
& \hspace{40mm}  i \in \mathcal V,\ p \in \mathcal P_t, \ t\in \mathcal T\label{eqn:mix-f-conv-multi}\\
& f_{ij,p}^{t}(l)\leq x_{ij,p}(l)B_{ij}, \  l=1,\cdots, L,\ (i,j)\in \mathcal E, \nonumber\\
& \hspace{32mm}  \ p \in \mathcal P_t, \ t\in \mathcal T \label{eqn:mix-f-x-cont-multi}\\
& \mathbf x_{s_pj}(l)=\mathbf e_p, \  l=1,\cdots, L,\ (s_p,j)\in \mathcal E, \ p\in \mathcal P\label{eqn:f-x-src-multi}\\
& \mathbf x_{ij}(l)=\vee_{k\in \mathcal I_i, m=1,\cdots, L} \beta_{kij}(m,l) \mathbf x_{ki}(m), \nonumber\\
&  \hspace{25mm}  l=1,\cdots, L,\ (i,j)\in \mathcal E, \ i \not\in\mathcal S\label{eqn:mix-x-inter-multi}\\
&x_{it,p}(l)=0, \  l=1,\cdots,L,\  i\in \mathcal I_t, \ p \not\in\mathcal P_t, \ t\in \mathcal T\label{eqn:mix-x-dest-multi}
\end{align}
where 
\begin{align}
\sigma_{i,p}^{t}=
\begin{cases}
R_p, & i=s_p\\
-R_p, &i=t\\
0, & \text{otherwise}
\end{cases}\quad  i \in \mathcal V,\ p \in \mathcal P_t, \ t\in \mathcal T.\label{eqn:sigma}
\end{align}
\label{Prob:new-low-cont-multi}
\end{Prob}

In the above formulation,\footnote{Note that \eqref{eqn:mix-x-multi} with $j=t$, \eqref{eqn:mix-f-conv-multi} with $i=t$, and \eqref{eqn:mix-f-x-cont-multi} with $j=t$ imply $\vee_{i\in \mathcal I_t, l=1,\cdots, L} x_{it,p}(l)=1$ for all  $p \in \mathcal P_t$, i.e., Condition 3) of Definition \ref{Def:feasibility-mixing},  where $t\in \mathcal T$.} 
$f_{ij,p}^{t}(l)\geq0$ can be interpreted as  the  rate of delivering flow $p\in \mathcal P_t$ to terminal $t\in \mathcal T$   over edge $(i,j)\in \mathcal E$  using $\mathbf x_{ij}(l)$,  and $ z_{ij}(l)$ denotes the transmission rate corresponding to $\mathbf x_{ij}(l)$ over edge $(i,j)\in \mathcal E$, where $l=1,\cdots,L$.  Problem \ref{Prob:new-low-cont-multi} is a mixed discrete-continuous optimization problem and is NP-complete in general. For notational simplicity, in this paper, we omit the conditions $\{z_{ij}:(i,j)\in \mathcal E\}$, $\{z_{ij}(l):l=1,\cdots, L,(i,j)\in \mathcal E\}$, $\{f_{ij,p}^{t}(l):l=1,\cdots, L, (i,j)\in \mathcal E,p\in\mathcal P_t, t\in \mathcal T\}$, $\{ x_{ij,p}(l):l=1,\cdots, L, (i,j)\in \mathcal E, p\in \mathcal P\}$, and $\{\beta_{kij}(l,m):l,m=1,\cdots, L, (k,i),(i,j)\in\mathcal E\}$  where there is no confusion.

\begin{Rem} [Problem \ref{Prob:new-low-cont-multi} with $L=1$ for Multicast]  For the multicast case (i.e., 
$\mathcal P_t=P$ for all $t\in \mathcal T$) and $L=1$,  the constraint in \eqref{eqn:mix-x-dest-multi} does not exist,  the constraint in \eqref{eqn:mix-f-z-sum-multi} can be satisfied by choosing $z_{ij}(1)=z_{ij}$, and  the constraint in \eqref{eqn:mix-f-x-cont-multi} is always satisfied  by choosing $\beta_{kij}(1,1)=1$ and  choosing $x_{ij,p}(1)$ according to \eqref{eqn:f-x-src-multi} and \eqref{eqn:mix-x-inter-multi}.   Therefore, in the  multicast case, Problem \ref{Prob:new-low-cont-multi} with $L=1$ for general connections reduces to the conventional minimum-cost   network coding design problem for the multicast case \cite{Lunetal06}.
\label{Rem:multi-L-1}
\end{Rem}

\begin{Rem} [Comparison with Intra-flow Coding] Problem \ref{Prob:new-low-cont-multi}  (with any $L\in\{1,\cdots, L_{\max}\}$) with an extra constraint $\sum_{p\in\mathcal P}x_{ij,p}(l)\in\{0,1\}$ for all $(i,j)\in \mathcal E$ and $l=1,\cdots, L$ is equivalent to a minimum-cost intra-flow coding problem. Thus, the minimum network cost of Problem \ref{Prob:new-low-cont-multi}  (with any $L\in\{1,\cdots, L_{\max}\}$) is no greater than the minimum costs for  intra-flow coding. \end{Rem}

\begin{Rem} [Comparison with Two-step Mixing]  Problem \ref{Prob:new-low-cont-multi} with $L=L_{\max}$ and $\beta_{kij}(l,m)=1$  instead of \eqref{eqn:mix-beta-multi}, is equivalent to the minimum-cost flow rate control  problem  in the second step of the two-step mixing approach  in \cite{Lun04networkcoding}. Thus,  the minimum network cost of Problem \ref{Prob:new-low-cont-multi} with $L=L_{\max}$ is no greater than the minimum cost of the  two-step mixing approach in \cite{Lun04networkcoding}.\end{Rem}

%
%
%

\begin{Exam}[Illustration of Linear Network Mixing] We illustrate a feasible mixing design (corresponding to a feasible solution) to   Problem \ref{Prob:new-low-cont-multi} with $L=2$ for the example in Fig.~\ref{Fig:para_cont}.  For ease of illustration,  in this example, we consider unit source rate and do not consider flow splitting and coding over time. For source edges (1,6), (1,4), (2,7), (2,4) and (3,4),  choose the global mixing vectors as follows: $\mathbf x_{16}(l)=\mathbf x_{14}(l)=(1,0,0)$, $\mathbf x_{24}(l)=\mathbf x_{27}(l)=(0,1,0)$ and $\mathbf x_{34}(l)=(0,0,1)$ for all $l=1,2$.  In addition,  choose the local coding coefficients as follows:   $\beta_{145}(1,1)=\beta_{245}(1,1)=\beta_{345}(1,2)=1$, $\beta_{145}(2,1)=\beta_{245}(2,1)=\beta_{345}(2,2)=0$,  $\beta_{145}(m,2)=\beta_{245}(m,2)=\beta_{345}(m,1)=0$ for all  $m=1,2$,  $\beta_{456}(1,1)=1$, $\beta_{456}(2,1)=\beta_{456}(1,2)=\beta_{456}(2,2)=0$, $\beta_{457}(1,1)=\beta_{457}(2,2)=1$ and  $\beta_{457}(1,2)=\beta_{457}(2,1)=0$. Therefore, for edges (4,5), (5,6) and (5,7) not outgoing from a source, the global mixing vectors are given by $\mathbf x_{45}(1)=(1,1,0)$, $\mathbf x_{45}(2)=(0,0,1)$, $\mathbf x_{56}(1)=(1,1,0)$, $\mathbf x_{56}(2)=(0,0,0)$, $\mathbf x_{57}(1)=(1,1,0)$ and $\mathbf x_{57}(2)=(0,0,1)$. On the other hand, flow paths  (sets of edge-mixing index pairs $((i,j),l)$ for which the rates of delivering  flows  are one) from the three sources, i.e., $\{((i,j),l):f_{ij,p}^t(l)=1,\  (i,j)\in \mathcal E, \ l=1,\cdots, L\}$ for all $p\in \mathcal P_t$ and $t\in \mathcal T$,   are illustrated using green, blue and pink curves in Fig.~\ref{Fig:para_cont}. Accordingly, choose the transmission rates as follows: 
$z_{ij}(1)=1$ and $z_{ij}(2)=0$ for all $(i,j)=(1,6), (1,4), (2,7), (2,4), (3,4)$, $z_{45}(1)=z_{45}(2)=z_{56}(1)=z_{57}(1)=z_{57}(2)=1$, $z_{56}(2)=0$, and $z_{ij}=z_{ij}(1)+z_{ij}(2) $ for all $(i,j)\in \mathcal E$.\label{example-mixing}
\end{Exam}

The following lemma shows the existence of a feasible linear network code corresponding to  Problem \ref{Prob:new-low-cont-multi}. 
\begin{Lem} Suppose Problem \ref{Prob:new-low-cont-multi} is feasible. Then, for each feasible solution, there exists  a feasible linear network code with a field size $F>T$ to deliver the desired flows to each terminal. \label{Lem:feasibility-new-low-opt-cont}
\end{Lem}

\begin{proof}Please refer to Appendix A.
\end{proof}

Note that a feasible linear network code can be obtained from a feasible linear network mixing design (a feasible solution to  Problem \ref{Prob:new-low-cont-multi}) using random linear network coding\cite{Hoetal06}, as illustrated in Section~\ref{subsec:coding-mixing}. 

\begin{Exam}[Illustration of Linear Network Coding] We illustrate how to obtain a feasible linear network code using random linear network coding, based on the feasible linear network mixing design illustrated in Example \ref{example-mixing}. In this example,  one local mixing coefficient (global mixing vector)  corresponds to one local coding coefficient (global coding vector).\footnote{Note that when flow splitting or coding over time happens, one local mixing coefficient (global mixing vector) may correspond to multiple local coding coefficients (global coding vectors). In this case, a linear network code can be designed in a similar way based on the sub-flows and sub-edges established in the proof of Lemma \ref{Lem:feasibility-new-low-opt-cont}.} For the source edges, choose the global coding vectors as follows: $\mathbf c_{ij}(l)=\mathbf x_{ij}(l)$ for all $(i,j)=(1,6), (1,4), (2,7), (2,4), (3,4)$ and  $l=1,2$. 
 For all $l,m=1,\cdots, L$ and $(k,i),(i,j)\in \mathcal E$, if $\beta_{kij}(l,m)=0$, choose $\alpha_{kij}(l,m)=0$; if $\beta_{kij}(l,m)=1$, choose $\alpha_{kij}(l,m)$ uniformly at random from $\mathcal F$. 
Therefore, for the edges  not outgoing from a source, the global coding vectors are given by $\mathbf c_{ij}(l)=\sum_{k\in \mathcal I_i, m=1,\cdots, L}\alpha_{kij}(m,l)\mathbf c_{ki}(m)$ for all $(i,j)=(4,5), (5,6),(5,7)$ and $l=1,\cdots, L$.
\end{Exam}

\subsection{Network Cost and Complexity Tradeoff}

 The design parameter $L$ in Problem \ref{Prob:new-low-cont-multi} determines the complexity and network cost tradeoff. First, 
we  illustrate the impact of  $L$ on the complexity of Problem \ref{Prob:new-low-cont-multi}. 
By \eqref{eqn:mix-beta-multi}, we know that for given $(k,i),(i,j)\in \mathcal E$, the cardinality of  $\{\beta_{kij}(l,m):l,m=1,\cdots, L\}$ is $L^2$. 
Since  $\sum_{(i,j)\in \mathcal E} O_j=\sum_{j\in \mathcal V}I_j O_j\leq \sum_{j\in \mathcal V}DO_j = D E$,  the cardinality of   $\{\beta_{kij}(l,m):l,m=1,\cdots, L,\ (k,i),(i,j)\in \mathcal E\}$ is smaller than or equal to $L^2DE$. 
Note that by \eqref{eqn:f-x-src-multi} and \eqref{eqn:mix-x-inter-multi}, $ \{ x_{ij,p}(l)\}$ can be  fully determined   by  $\{\beta_{kij}(l,m)\}$. Therefore, the cardinality of  the discrete variables  $ \{ x_{ij,p}(l)\}$ and $\{\beta_{kij}(l,m)\}$  of Problem \ref{Prob:new-low-cont-multi} is $L^2DE$, which increases as $L$ increases.   

Next, we discuss the impact of   $L$ on the  network cost.

\begin{Lem}
If Problem \ref{Prob:new-low-cont-multi} is feasible for design parameter $L$, then Problem \ref{Prob:new-low-cont-multi} is feasible for design parameter $L+1$ and $ {U_{x}^*}(L+1)\leq  {U_{x}^*}(L)$, where $L=1,\cdots, L_{\max}-1$. \label{Lem:comp-L-cont-multi}
\end{Lem}

\begin{proof} Given a feasible solution to Problem \ref{Prob:new-low-cont-multi} with design parameter $L$, by setting variables w.r.t. index $l=L+1$ or $m=L+1$ to be zero, we can easily construct a feasible solution to Problem \ref{Prob:new-low-cont-multi} with design parameter $L+1$. This feasible solution corresponds to the same network cost as the one with design parameter $L$. But the network cost with design parameter $L+1$ can be further optimized by solving Problem \ref{Prob:new-low-cont-multi} with design parameter $L+1$. Therefore, we can show $ {U_{x}^*}(L+1)\leq  {U_{x}^*}(L)$ for all $L=1,\cdots, L_{\max}-1$. 
\end{proof}

By Lemma \ref{Lem:comp-L-cont-multi}, we know that  the network cost ${U_{x}^*}(L)$ is non-increasing w.r.t. $L$.  
This can also be understood from the example in Fig. \ref{Fig:para_cont}. Note that by Condition 3) in Definition \ref{Def:feasibility-mixing}, flow 3 is not allowed to be mixed with flow 1 and flow 2 on their paths to terminal $t_1$. When $L=1<L_{\max}$, flow 3 cannot be delivered over edge $(4,5)$ to terminal $t_2$ using feasible mixing. In other words, Problem \ref{Prob:new-low-cont-multi} with $L=1$ is not feasible (i.e., of infinite network cost). However, when $L=2=L_{\max}$, flow 3 can be delivered to  terminal $t_2$ without mixing with flow 1 and flow 2 over edge $(4,5)$, e.g., using global mixing vectors $\mathbf x_{45}(1)=(1,1,0)$ and $\mathbf x_{45}(2)=(0,0,1)$ over edge $(4,5)$. In other words, Problem \ref{Prob:new-low-cont-multi} with $L=2$ is feasible (i.e., of finite network cost). Thus, we can see the impact of $L$ on the  network cost shown in Lemma \ref{Lem:comp-L-cont-multi}.

%

\section{Alternative Formulation with Discrete Mixing}\label{subsec:alt-cont-mix-dis}

Problem \ref{Prob:new-low-cont-multi} is a mixed discrete-continuous optimization problem with two main challenges. One is the choice of the   network mixing  coefficients  (discrete variables), and the other is the choice of the flow rates (continuous variables).  In this section, we first propose an equivalent alternative formulation of Problem \ref{Prob:new-low-cont-multi} which  naturally subdivides Problem \ref{Prob:new-low-cont-multi} according to these two aspects.  Then, we propose a distributed algorithm to solve it.

\subsection{Alternative Formulation}

Problem \ref{Prob:new-low-cont-multi} is equivalent to the following problem.

\begin{Prob}[Equivalent Problem of Problem \ref{Prob:new-low-cont-multi}]
\begin{align}
 {U_{x}^*}(L)=&\min_{\{x_{ij,p}(l)\}\in \boldsymbol{\mathcal M}(L)}\quad   U_{x}^*(\{x_{ij,p}(l)\})\nonumber
\end{align}
where $U_{x}^*(\{x_{ij,p}(l)\})$ and $\boldsymbol{\mathcal M}(L)$ are given by the following two subproblems. \label{Prob:equ-mix}
\end{Prob}

\begin{Subp} [Subproblem of Problem \ref{Prob:equ-mix}: Flow Optimization] For given $\{ x_{ij,p}(l)\}$, we have: 
\begin{align}
U_{x}^*(\{ x_{ij,p}(l)\})&=\min_{\{z_{ij}\},\{z_{ij}(l)\} ,\{f_{ij,p}^{t}(l)\}}\quad  \sum_{(i,j)\in \mathcal E} U_{ij}(z_{ij})\nonumber\\
s.t.\quad & \eqref{eqn:mix-z-multi}, \eqref{eqn:mix-f-multi}, \eqref{eqn:mix-f-z-multi}, \eqref{eqn:mix-f-z-sum-multi}, \eqref{eqn:mix-f-conv-multi},\eqref{eqn:mix-f-x-cont-multi}\nonumber
\end{align}
 \label{Prob:flowopt}
\end{Subp}

\begin{Subp} [Subproblem of Problem \ref{Prob:equ-mix}: Feasible Mixing]
Find the set $\boldsymbol{\mathcal M}(L)\triangleq\{\{ x_{ij,p}(l)\}:\eqref{eqn:mix-x-multi}, \eqref{eqn:mix-beta-multi},\eqref{eqn:f-x-src-multi},\eqref{eqn:mix-x-inter-multi},\eqref{eqn:mix-x-dest-multi},\eqref{eqn:mix-x-dest-multi-add}\}$ of feasible $\{ x_{ij,p}(l)\}$, where \eqref{eqn:mix-x-dest-multi-add} is given by:
\begin{align}
&\vee_{i\in \mathcal I_t, l=1,\cdots, L} x_{it,p}(l)=1,  \ p \in \mathcal P_t, \ t\in \mathcal T.\label{eqn:mix-x-dest-multi-add}
\end{align}
\label{Prob:feasible-set}
\end{Subp}

Note that for given $\{ x_{ij,p}(l)\}$, Subproblem \ref{Prob:flowopt} is a convex optimization problem (involving continuous flow rates) and hence has  polynomial-time complexity.  On the other hand, Subproblem \ref{Prob:feasible-set} is a discrete feasibility problem (involving discrete mixing coefficients) and is NP-complete in general. Therefore, Problem \ref{Prob:equ-mix} is still a mixed discrete-continuous optimization problem and is NP-complete in general. 

\subsection{Distributed Solution}\label{subsec:equ-mix-dist}

In this part, we develop a distributed algorithm to solve Problem \ref{Prob:equ-mix} by solving Subproblem \ref{Prob:flowopt}  and Subproblem \ref{Prob:feasible-set}, respectively, in a distributed manner. First, we consider  Subproblem \ref{Prob:feasible-set}.
Subproblem \ref{Prob:feasible-set} can be treated as a CSP and solved distributively using  clause partition and the Communication-Free Learning (CFL) algorithm from \cite{cfl}.  While CSPs
are NP-complete in general, CFL provides  a probabilistic distributed iterative 
algorithm with almost sure convergence in finite time.  
Specifically, $\{ x_{ij,p}(l)\}\cup\{\beta_{kij}(l,m)\}$ can be treated as the variables of the CSP. 
$\{0,1\}$ can be treated as the finite set  of the CSP. 
From \eqref{eqn:mix-x-inter-multi}, we have an equivalent constraint purely on $\{ x_{ij,p}(l)\}$, i.e.,
\begin{align}
&\exists \ \beta_{kij}(m,l)\in \{0,1\} \ \forall   k\in \mathcal I_i, m=1,\cdots, L,\nonumber\\
& \text{s.t.}\  \mathbf x_{ij}(l)=\vee_{k\in \mathcal I_i, m=1,\cdots, L} \beta_{kij}(m,l) \mathbf x_{ki}(m),  \nonumber\\
&\hspace{10mm} l=1,\cdots, L, \ (i,j)\in \mathcal E, \ i \not\in\mathcal S.\label{eqn:mix-x-inter-multi-exist}
\end{align}
In the following, we shall only consider solving for the  variables $\{ x_{ij,p}(l)\}$ of the CSP in a distributed way using clause partition and CFL, as 
$\{\beta_{kij}(l,m)\}$ can be obtained from feasible $\{ x_{ij,p}(l)\}$ by \eqref{eqn:f-x-src-multi} and \eqref{eqn:mix-x-inter-multi}.  In addition, we directly choose $\mathbf x_{s_pj}(l)=\mathbf e_p$ for all $l=1,\cdots, L$, $(s_p,j)\in \mathcal E $ and  $p\in \mathcal P$ according to \eqref{eqn:f-x-src-multi}.

For notational simplicity, we write the clauses for $\{ x_{ij,p}(l)\}$ in a more compact form as follows:
\begin{align}
&\phi^{x,l}_{ij,p}\Big( \mathbf x_{ij}(l), \{\mathbf x_{ki}(m):m=1,\cdots, L, k\in\mathcal I_i\},\nonumber\\
&\quad \quad \ \{\mathbf x_{kj}(m):m=1,\cdots, L, k\in\mathcal I_j,j\in\mathcal T\}\Big)\nonumber\\
\triangleq&\begin{cases}
1, &\text{if $j\not\in\mathcal T$, \eqref{eqn:mix-x-inter-multi-exist} holds}\\
1, &\text{if $j\in\mathcal T$ and $p\in\mathcal P_j$, \eqref{eqn:mix-x-inter-multi-exist} and   \eqref{eqn:mix-x-dest-multi-add}
 hold}\\
 1, &\text{if $j\in\mathcal T$ and $p\not\in\mathcal P_j$, \eqref{eqn:mix-x-inter-multi-exist}  and  \eqref{eqn:mix-x-dest-multi} hold}\\
0, & \text{otherwise}
\end{cases}\nonumber\\
& \quad l=1,\cdots,L,\ (i,j)\in \mathcal E,\ p\in\mathcal P,\ i\not\in S.\label{eqn:phi-x-multi}
\end{align}
Note that, when $j\not\in\mathcal T$, $\{\mathbf x_{kj}(m):m=1,\cdots, L,k\in\mathcal I_j,j\in\mathcal T\}=\emptyset$ and we ignore it in the clause $\phi^{x,l}_{ij,p}(\cdot)$.
For \eqref{eqn:mix-x-dest-multi-add} and \eqref{eqn:mix-x-dest-multi} in clause $\phi^{x,l}_{ij,p}(\cdot)$, we use $j$ as the terminal index instead of $t$. It can be seen that the constraints in 
 \eqref{eqn:mix-x-inter-multi} (i.e., \eqref{eqn:mix-x-inter-multi-exist}), \eqref{eqn:mix-x-dest-multi} and \eqref{eqn:mix-x-dest-multi-add} are considered in clause $\phi^{x,l}_{ij,p}(\cdot)$.  In addition, the constraint in \eqref{eqn:f-x-src-multi} is considered when choosing $\mathbf x_{s_pj}(l)=\mathbf e_p$ for all $l=1,\cdots,L$, $(s_p,j)\in \mathcal E $ and  $p\in \mathcal P$. Therefore, the CSP has considered all the constraints in Subproblem \ref{Prob:feasible-set}.

We now construct the clause partition of Subproblem \ref{Prob:feasible-set}. Specifically, the set of clauses variable $x_{ij,p}(l)$ participates in is as follows: 
\begin{align}
\Phi^{x,l}_{ij,p}\triangleq&\left\{\phi^{x,l}_{ij,p}\right\}\cup\left\{\phi^{x,m}_{jk,p}:m=1,\cdots, L, k\in\mathcal O_j\right\}\nonumber\\
&\cup\left\{\phi^{x,m}_{kj,p}:m=1,\cdots, L, k\in\mathcal I_j,j\in\mathcal T\right\}\nonumber\\
&\quad  l=1,\cdots,L,\ (i,j)\in \mathcal E,\ p\in\mathcal P,\ i\not\in S.\label{eqn:phi-x-part-multi}
\end{align}
Note that, when $j\not\in\mathcal T$, $\left\{\phi^{x,m}_{kj,p}:m=1,\cdots, L,k\in\mathcal I_j,j\in\mathcal T\right\}=\emptyset$ and we ignore it in $\Phi^{x,l}_{ij,p}$ in \eqref{eqn:phi-x-part-multi}. 

Based on the clause partition, a feasible $\{x_{ij,p}(l)\}\in \boldsymbol{\mathcal M}(L)$ to Subproblem \ref{Prob:feasible-set}  can be found distributively using the  probabilistic distributed iterative CFL algorithm \cite[Algorithm 1]{cfl}.  Specifically, for all $(i,j)\in \mathcal E$, $p\in\mathcal P$ and $l=1,\cdots, L$,  in each iteration, each node $i$ realizes a Bernoulli random variable selecting  $x_{ij,p}(l)$; messages on  $\{x_{ij,p}(l)\}$ are   passed between adjacent nodes for each node $i$ to evaluate its related clauses in \eqref{eqn:phi-x-part-multi};  based on the evaluation, each node $i$ updates the distribution of the Bernoulli  random variable selecting  $x_{ij,p}(l)$.  
Given a feasible $\{x_{ij,p}(l)\}\in\boldsymbol{ \mathcal M}(L)$ obtained by CFL,   Subproblem \ref{Prob:flowopt} is convex and can be solved distributively using standard convex decomposition. We omit the details here due to the page limitation.
 
Now, we can develop  a distributed algorithm to solve  Problem \ref{Prob:equ-mix} based on CFL and convex decomposition, as briefly illustrated in Algorithm \ref{alg:equ-mix}.\footnote{
In Step 3, CFL is run for a sufficiently long time. Step 4 (Step 6) can be implemented  with a master node obtaining the network convergence information of CFL (network cost) from all nodes or with all nodes computing the average convergence indicator of CFL (average network cost) locally via a gossip algorithm.}

\begin{algorithm}[h]
\caption{Algorithm for Problem \ref{Prob:equ-mix}}
\begin{algorithmic}[1]
\STATE \textbf{initialize}    $n=1$ and $U_1=+\infty$. 
\LOOP
\STATE Run CFL  to  the CSP corresponding to Subproblem \ref{Prob:feasible-set}.
 \IF{\text{the CFL finds a feasible solution}} 
 \STATE For the obtained feasible solution to Subproblem \ref{Prob:feasible-set}, solve Subproblem \ref{Prob:flowopt} distributively using convex decomposition.  Let $\bar U_n$ denote the corresponding network cost. 
 \IF{\text{$\bar U_n<U_n$}}
 \STATE set $U_{n+1}=\bar U_n$ and $n=n+1$.
\ENDIF
 \ENDIF 
\ENDLOOP
\end{algorithmic}\label{alg:equ-mix}
\end{algorithm}

Based on the convergence result of CFL \cite[Corollary 2]{cfl}, we can easily see that $U_n\to U_x^*(L)$ almost surely as $n\to \infty$, if Problem \ref{Prob:equ-mix}  is feasible.




\section{Alternative Formulation with Continuous Mixing}\label{subsec:alt-cont-mix-cont}

The complexity of solving Problem \ref{Prob:equ-mix}  mainly lies in solving for the  network mixing coefficients (discrete variables) in Subproblem \ref{Prob:feasible-set}.  In this section, we first propose an equivalent alternative formulation of Problem \ref{Prob:new-low-cont-multi}  (Problem \ref{Prob:equ-mix}) with continuous mixing. Then, we elaborate on some distributed algorithms to solve it.

\subsection{Alternative Formulation} 

Problem \ref{Prob:new-low-cont-multi}  (Problem \ref{Prob:equ-mix})  is a mixed discrete-continuous optimization problem.  Applying continuous relaxation to \eqref{eqn:mix-x-multi} and \eqref{eqn:mix-beta-multi}  and manipulating  \eqref{eqn:mix-x-inter-multi}, we obtain the following continuous optimization problem. 

\begin{Prob} [Continuous  Formulation of Problem \ref{Prob:new-low-cont-multi}]
\begin{align}
 {\bar U_{x}^*}(L)&=\min_{\substack{\{z_{ij}\},\{z_{ij}(l)\} \{f_{ij,p}^{t}(l)\}\\ \{\bar{x}_{ij,p}(l)\},\{\bar\beta_{kij}(l,m)\}}}\quad  \sum_{(i,j)\in \mathcal E} U_{ij}(z_{ij})\nonumber\\
s.t.\  
&\eqref{eqn:mix-z-multi},\eqref{eqn:mix-f-multi}, \eqref{eqn:mix-f-z-multi},\eqref{eqn:mix-f-z-sum-multi}, \eqref{eqn:mix-f-conv-multi},\eqref{eqn:f-x-src-multi},\eqref{eqn:mix-x-dest-multi}\nonumber\\
&\bar x_{ij,p} (l)\in [0,1], \   l=1,\cdots, L,\ (i,j)\in \mathcal E,\ p\in \mathcal P\label{eqn:mix-x-multi-cont}\\
& \bar{\beta}_{kij}(l,m)\in [0,1], \  l,m=1,\cdots, L, \ (k,i),  (i,j) \in \mathcal E \label{eqn:mix-beta-multi-cont}\\
& f_{ij,p}^{t}(l)\leq \bar x_{ij,p}(l)B_{ij},  \  l=1,\cdots, L,\ (i,j)\in \mathcal E,\nonumber\\
& \hspace{32mm}   p \in \mathcal P_t, \ t\in \mathcal T\label{eqn:mix-f-x-cont-multi-cont}\\
&\bar x_{ij,p}(m)\geq  \bar \beta_{kij}(l,m) \bar x_{ki,p}(l), \ k\in \mathcal I_i, \ l=1,\cdots, L\nonumber\\
&  \hspace{2mm}\quad   m=1,\cdots, L,\  \mathcal I_i \neq \emptyset, \ (i,j)\in \mathcal E, p\in \mathcal P\label{eqn:mix-x-inter-multi-cont-1}\\
&\bar x_{ij,p}(m)\leq \sum_{k\in \mathcal I_i, l=1,\cdots, L} \bar\beta_{kij}(l,m) \bar x_{ki,p}(l), \nonumber\\
&  \hspace{2mm}\quad   m=1,\cdots, L,\  \mathcal I_i \neq \emptyset, \ (i,j)\in \mathcal E, p\in \mathcal P\label{eqn:mix-x-inter-multi-cont-2}
\end{align}
\label{Prob:new-low-cont-multi-cont}
\end{Prob}

Note that Constraints \eqref{eqn:mix-x-multi-cont} and \eqref{eqn:mix-beta-multi-cont}  in Problem \ref{Prob:new-low-cont-multi-cont}
can be treated as the continuous relaxation of Constraints \eqref{eqn:mix-x-multi} and \eqref{eqn:mix-beta-multi}  in Problem \ref{Prob:new-low-cont-multi}. 
Constraint \eqref{eqn:mix-f-x-cont-multi-cont} in Problem \ref{Prob:new-low-cont-multi-cont} corresponds to 
 Constraint \eqref{eqn:mix-f-x-cont-multi} in Problem \ref{Prob:new-low-cont-multi}. 
Constraints \eqref{eqn:mix-x-inter-multi-cont-1} and \eqref{eqn:mix-x-inter-multi-cont-2} in Problem \ref{Prob:new-low-cont-multi-cont}  can be treated as the  continuous counterpart of Constraint \eqref{eqn:mix-x-inter-multi} in Problem \ref{Prob:new-low-cont-multi}. 
The following lemma shows the relationship between Problem \ref{Prob:new-low-cont-multi} (mixed discrete-continuous optimization problem) and Problem \ref{Prob:new-low-cont-multi-cont} (continuous optimization problem).

\begin{Lem} [Relationship between  Problem \ref{Prob:new-low-cont-multi} and Problem \ref{Prob:new-low-cont-multi-cont}] (i) If
$\{z_{ij}\},\{z_{ij}(l)\} ,\{f_{ij,p}^{t}(l)\}, \{ x_{ij,p}(l)\},\{\beta_{kij}(l,m)\}$ is a feasible solution to  Problem \ref{Prob:new-low-cont-multi}, then $\{z_{ij}\},\{z_{ij}(l)\} ,\{f_{ij,p}^{t}(l)\}, \{\bar{x}_{ij,p}(l)\},\{\bar \beta_{kij}(l,m)\}$ is a feasible solution to Problem \ref{Prob:new-low-cont-multi-cont}, where $\bar x_{ij,p}(l)=x_{ij,p}(l)$ and $\bar \beta_{kij}(l,m)=\beta_{kij}(l,m)$; if $\{z_{ij}\},\{z_{ij}(l)\} ,\{f_{ij,p}^{t}(l)\}, \{\bar{x}_{ij,p}(l)\},\{\bar \beta_{kij}(l,m)\}$ is a feasible solution to Problem \ref{Prob:new-low-cont-multi-cont}, then $\{z_{ij}\},\{z_{ij}(l)\}, \{f_{ij,p}^{t}(l)\}, \{x_{ij,p}(l)\},\{\beta_{kij}(l,m)\}$ is a feasible solution to  Problem \ref{Prob:new-low-cont-multi}, where $x_{ij,p}(l)=\lceil\bar x_{ij,p}(l)\rceil$ and $\beta_{kij}(l,m)=\lceil\bar\beta_{kij}(l,m)\rceil$. 
(ii)  The feasibilities of Problem \ref{Prob:new-low-cont-multi} and Problem \ref{Prob:new-low-cont-multi-cont} imply each other. 
(iii) 
The optimal values of Problem \ref{Prob:new-low-cont-multi} and Problem \ref{Prob:new-low-cont-multi-cont} are the same, i.e., $ { U_{x}^*}(L)= {\bar U_{x}^*}(L)$. 
\label{Lem:relation-cont-relax}
\end{Lem}

\begin{proof} Please refer to Appendix B.
\end{proof}


By Lemma \ref{Lem:relation-cont-relax},  solving Problem \ref{Prob:new-low-cont-multi}   is equivalent to solving Problem \ref{Prob:new-low-cont-multi-cont}.

 \subsection{Distributed Solution} 
 
 Problem \ref{Prob:new-low-cont-multi-cont} is a (pure) continuous optimization problem. It   is not convex due to Constraints \eqref{eqn:mix-x-inter-multi-cont-1} and \eqref{eqn:mix-x-inter-multi-cont-2}.  Several penalty methods \cite{Bertsekasbooknonlinear:99} can be applied to find a local minimum of Problem \ref{Prob:new-low-cont-multi-cont}  with polynomial-time complexity.  Those methods can also be implemented in a distributed manner using standard decomposition.  
 On the other hand, by weak duality \cite{Bertsekasbooknonlinear:99}, dual method can be applied to find a lower bound of the global minimum value of Problem \ref{Prob:new-low-cont-multi-cont}. The difference between an obtained local minimum and this lower bound can serve as an upper bound on the performance  gap between the local minimum and the global minimum  of Problem \ref{Prob:new-low-cont-multi-cont}. 
 
 \section{Conclusion}
In this paper, we considered  linear network code constructions for general connections of continuous flows to minimize the total cost of  edge use based on mixing.  
 To solve the minimum-cost network coding design problem, we proposed two equivalent alternative formulations with discrete mixing and continuous mixing, respectively, and developed distributed algorithms to solve them. 
Our approach allows fairly general coding across flows and guarantees no greater cost than any solution without  inter-flow network coding.

\section*{Appendix A: Proof of Lemma \ref{Lem:feasibility-new-low-opt-cont}}

First, we consider $L=1$. We omit the index terms $(1)$ and $(1,1)$ behind the variables for notational simplicity. Let $\{z_{ij}\}$, $\{ x_{ij,p}\}$, $\{\beta_{kij}\}$  and $\{f_{ij,p}^{t}\}$ denote a feasible solution to Problem \ref{Prob:new-low-cont-multi}.   We shall extend the proof of Lemma 1 in \cite{CUI2015ISITreport}  for the integer flows ($f_{ij,p}^{t} \in\{0,1\}$) and unit source rates ($R_p=1$) with one global coding vector over each edge ($z_{i,j}\in\{0,1\}$)  to the general continuous flows ($f_{ij,p}^{t} \in[ 0,B_{ij}]$) and source rates ($R_p\in \mathbb R^+$) with multiple global coding vectors ($z_{i,j}\in [0,B_{ij}]$) over each edge. In the general case, we code over time $n\geq 1$. For all $p\in\mathcal P$, convert  source  $p$ with source rate $R_p$ over time $n$ to $\lfloor nR_p\rfloor$ unit rate sub-sources $p_1,\cdots, p_{\lfloor nR_p\rfloor}$.  For each edge $(i,j)\in \mathcal E$, allow the total number of   the sub-flows of flow $p\in \mathcal P_t$ to terminal $t\in \mathcal T$ to be fewer than or equal to $\lceil nf_{ij,p}^{t}\rceil $. Therefore, the flow path of flow $p$ can be decomposed  into $\lfloor nR_p\rfloor$ unit rate sub-flow paths $p_1,\cdots, p_{\lfloor nR_p\rfloor}$  from source $p\in \mathcal P_t$ to terminal $t\in \mathcal T$. The sum rate of unit rate sub-flows of flow $p$ over edge $(i,j)\in \mathcal E$ is less than or equal to $\lceil nf_{ij,p}^{t}\rceil $. The sum rate of unit rate sub-flows of all the flows  over edge $(i,j)$ is less than or equal to $\bar z_{ij}=\max_{t\in \mathcal T}\sum_{p\in \mathcal P_t}\lceil nf_{ij,p}^{t}\rceil  $.   Decompose edge $(i,j)$ into $\bar z_{ij}$ sub-edges. Let sub-flows to terminal $t$ pass different sub-edges, i.e., each sub-edge transmit at most one sub-flow to terminal $t$. We have now reduced the general case to the special case considered in Lemma 1 in \cite{CUI2015ISITreport}. Therefore, we can show that there exists a feasible linear network code over time $n$. The associated average sum transmission rate over edge $(i,j)$ is $\bar z_{ij}/n$. Note that $\bar z_{ij}/n-z_{ij}/n\leq P/n$. Therefore, this code design can achieve the minimum cost $U_{x}^*(1)$ by taking $n$ arbitrarily large. 

When $L>1$, we can convert  each edge $(i,j)\in \mathcal E$ into $L$  edges. Then, we can apply the above proof for $L=1$ to the equivalent constructed network.

\section*{Appendix B: Proof of Lemma \ref{Lem:relation-cont-relax}}
It is obvious that (i) implies (ii). Next, we show that (i) implies (iii). Suppose (i) holds, which  indicates that  each $\{z_{ij}\}$ associated with a feasible solution to Problem \ref{Prob:new-low-cont-multi} is also associated with a feasible solution to Problem \ref{Prob:new-low-cont-multi-cont}, and vice versa. In addition,  $\{z_{ij}\}$ fully determines $\sum_{(i,j)\in \mathcal E} U_{ij}(z_{ij})$. Thus,  the set of feasible network costs to Problem \ref{Prob:new-low-cont-multi} is the same as that to Problem \ref{Prob:new-low-cont-multi-cont}, implying the optimal values of the two problems are the same. Therefore, we can show that (i) implies (iii).  Thus, to show Lemma \ref{Lem:relation-cont-relax},  it is sufficient to show (i). Note that in the proof, we only need to consider the different constrains between Problem \ref{Prob:new-low-cont-multi} and Problem \ref{Prob:new-low-cont-multi-cont}.

To show (i), we first show  that when $x_{ij,p}(l)\in\{0,1\}$ and $ \beta_{kij}(l,m)\in\{0,1\}$,  Constraint \eqref{eqn:mix-x-inter-multi}  is equivalent to the following two constraints in \eqref{eqn:mix-x-inter-multi-1} and \eqref{eqn:mix-x-inter-multi-2}. 
\begin{align}
& x_{ij,p}(m)\geq   \beta_{kij}(l,m)  x_{ki,p}(l), \ k\in \mathcal I_i, \ l=1,\cdots, L\nonumber\\
&  \hspace{2mm}\quad   m=1,\cdots, L,\  \mathcal I_i \neq \emptyset, \ (i,j)\in \mathcal E, p\in \mathcal P\label{eqn:mix-x-inter-multi-1}\\
& x_{ij,p}(m)\leq \sum_{k\in \mathcal I_i, l=1,\cdots, L} \beta_{kij}(l,m) x_{ki,p}(l), \nonumber\\
&  \hspace{2mm}\quad   m=1,\cdots, L,\  \mathcal I_i \neq \emptyset, \ (i,j)\in \mathcal E, p\in \mathcal P\label{eqn:mix-x-inter-multi-2}
\end{align}
Note that Constraints \eqref{eqn:mix-x-inter-multi}, \eqref{eqn:mix-x-inter-multi-1} and \eqref{eqn:mix-x-inter-multi-2} are for all $m=1,\cdots, L,\  \mathcal I_i \neq \emptyset, \ (i,j)\in \mathcal E$ and $p\in \mathcal P$. Thus, we prove this equivalence by considering the following two cases for any $ m=1,\cdots, L,\  \mathcal I_i \neq \emptyset, \ (i,j)\in \mathcal E$  and $p\in \mathcal P$.  First, consider the case where $\beta_{kij}(l,m)  x_{ki,p}(l)=0$ for all $k\in \mathcal I_i$ and $ l=1,\cdots, L$. Constraint \eqref{eqn:mix-x-inter-multi}  implies that $x_{ij,p}(m)=0$, and Constraints \eqref{eqn:mix-x-inter-multi-1} and \eqref{eqn:mix-x-inter-multi-2}  also imply that $x_{ij,p}(m)=0$. Second, consider 
  the case where there exists at least one pair $(k,l)$, where $k\in \mathcal I_i$ and $ l=1,\cdots, L$, such that $\beta_{kij}(l,m)  x_{ki,p}(l)=1$. Constraint \eqref{eqn:mix-x-inter-multi}  implies that $x_{ij,p}(m)=1$, and Constraints \eqref{eqn:mix-x-inter-multi-1} and \eqref{eqn:mix-x-inter-multi-2}  also imply that $x_{ij,p}(m)=1$. Note that under the integer constraints $x_{ij,p}(l)\in\{0,1\}$ and $ \beta_{kij}(l,m)\in\{0,1\}$, the above two cases are the only two possible cases. Therefore, we can show  Constraint \eqref{eqn:mix-x-inter-multi}  is equivalent to Constraints \eqref{eqn:mix-x-inter-multi-1} and \eqref{eqn:mix-x-inter-multi-2}. 

Next, we show that the first statement of (i) holds.  Suppose $\{z_{ij}\},\{z_{ij}(l)\}, \{f_{ij,p}^{t}(l)\}, \{x_{ij,p}(l)\},\{\beta_{kij}(l,m)\}$ is a feasible solution to  Problem \ref{Prob:new-low-cont-multi}. Let $\bar x_{ij,p}(l)=x_{ij,p}(l)\in\{0,1\}$ for all  $ l=1,\cdots, L,  \ (i,j)\in \mathcal E$ and $p\in \mathcal P$,  and $\bar \beta_{kij}(l,m)=\beta_{kij}(l,m)\in\{0,1\}$ for all $k\in \mathcal I_i, \  \mathcal I_i \neq \emptyset, \ (i,j)\in \mathcal E$ and $ l,m=1,\cdots, L $. Since Constraints \eqref{eqn:mix-x-multi-cont}, \eqref{eqn:mix-beta-multi-cont} and \eqref{eqn:mix-f-x-cont-multi-cont} in Problem \ref{Prob:new-low-cont-multi-cont}
can be treated as the continuous relaxation of Constraints \eqref{eqn:mix-x-multi}, \eqref{eqn:mix-beta-multi} and \eqref{eqn:mix-f-x-cont-multi} in Problem \ref{Prob:new-low-cont-multi}, 
$\{f_{ij,p}^{t}(l)\}, \{\bar{ x}_{ij,p}(l)\},\{\bar \beta_{kij}(l,m)\}$ satisfies Constraints \eqref{eqn:mix-x-multi-cont}, \eqref{eqn:mix-beta-multi-cont} and \eqref{eqn:mix-f-x-cont-multi-cont}. In addition, since Constraint \eqref{eqn:mix-x-inter-multi}  is equivalent to Constraints \eqref{eqn:mix-x-inter-multi-1} and \eqref{eqn:mix-x-inter-multi-2}, and Constraints \eqref{eqn:mix-x-inter-multi-cont-1}  and \eqref{eqn:mix-x-inter-multi-cont-2} can be treated as the continuous relaxation of Constraints \eqref{eqn:mix-x-inter-multi-1} and \eqref{eqn:mix-x-inter-multi-2}, $\{\bar{x}_{ij,p}(l)\},\{\bar \beta_{kij}(l,m)\}$ satisfies Constraints \eqref{eqn:mix-x-inter-multi-cont-1} and \eqref{eqn:mix-x-inter-multi-cont-2}. 
 Therefore, we can show $\{z_{ij}\},\{z_{ij}(l)\}, \{f_{ij,p}^{t}(l)\}, \{\bar{x}_{ij,p}(l)\},\{\bar \beta_{kij}(l,m)\}$ is a feasible solution to Problem \ref{Prob:new-low-cont-multi-cont}. 

Finally, we show that the second statement of (i) holds.  Suppose $\{z_{ij}\},\{z_{ij}(l)\},\{f_{ij,p}^{t}(l)\}, \{\bar{  x}_{ij,p}(l)\},\{\bar \beta_{kij}(l,m)\}$ is a feasible solution to Problem \ref{Prob:new-low-cont-multi-cont}. Let $x_{ij,p}(l)=\lceil\bar x_{ij,p}(l)\rceil$  for all  $ l=1,\cdots, L,  \ (i,j)\in \mathcal E$ and $p\in \mathcal P$, and $\beta_{kij}(l,m)=\lceil\bar\beta_{kij}(l,m)\rceil$ for all $k\in \mathcal I_i, \  \mathcal I_i \neq \emptyset, \ (i,j)\in \mathcal E$ and $ l,m=1,\cdots, L $.
In other words, if $\bar x_{ij,p}(l)=0$ ($\bar\beta_{kij}(l,m)=0$), then $ x_{ij,p}(l)=0$ ($\beta_{kij}(l,m)=0$); if $\bar x_{ij,p}(l,m)\in (0,1]$ ($\bar\beta_{kij}(l,m)\in(0,1]$), then $ x_{ij,p}(l)=1$ ($\beta_{kij}(l,m)=1$). It is obvious that $\{f_{ij,p}^{t}(l)\}, \{ x_{ij,p}(l)\},\{\beta_{kij}(l,m)\}$ satisfies Constraints \eqref{eqn:mix-x-multi}, \eqref{eqn:mix-beta-multi} and \eqref{eqn:mix-f-x-cont-multi}. It remains to show $\{x_{ij,p}(l)\},\{ \beta_{kij}(l,m)\}$ satisfies Constraint \eqref{eqn:mix-x-inter-multi}. Note that Constraint \eqref{eqn:mix-x-inter-multi} is for all $m=1,\cdots, L,\  \mathcal I_i \neq \emptyset, \ (i,j)\in \mathcal E$ and $ p\in \mathcal P$. 
Thus, similarly, we prove this result by considering the following two cases for any $ m=1,\cdots, L,\  \mathcal I_i \neq \emptyset, \ (i,j)\in \mathcal E$ and $p\in \mathcal P$.  First, consider the case where $\bar\beta_{kij}(l,m) \bar x_{ki,p}(l)=0$ for all $k\in \mathcal I_i$ and $ l=1,\cdots, L$.  Constraints \eqref{eqn:mix-x-inter-multi-cont-1} and \eqref{eqn:mix-x-inter-multi-cont-2}  imply that $\bar x_{ij,p}(m)=0$, and hence, we have $x_{ij,p}(m)=\lceil\bar x_{ij,p}(m)\rceil=0$. In addition, $\bar\beta_{kij}(l,m) \bar x_{ki,p}(l)=0$ for all $k\in \mathcal I_i$ and $ l=1,\cdots, L$ also implies $\beta_{kij}(l,m)  x_{ki,p}(l)=\lceil \bar\beta_{kij}(l,m)\rceil \lceil \bar  x_{ki,p}(l)\rceil=0$ for all $k\in \mathcal I_i$ and $ l=1,\cdots, L$. Thus, in this case, we can show $\{x_{ij,p}(l)\},\{ \beta_{kij}(l,m)\}$ satisfies Constraint  \eqref{eqn:mix-x-inter-multi}. Second, consider 
 the case where there exists at least one pair $(k,l)$, where $k\in \mathcal I_i$ and $ l=1,\cdots, L$, such that $\bar\beta_{kij}(l,m) \bar x_{ki,p}(l)\in (0,1]$. Constraints \eqref{eqn:mix-x-inter-multi-cont-1} and \eqref{eqn:mix-x-inter-multi-cont-2}  together with Constraints \eqref{eqn:mix-x-multi-cont} and \eqref{eqn:mix-beta-multi-cont}
  imply that $\bar x_{ij,p}(m)\in(0,1]$, and hence, we have $x_{ij,p}(m)=\lceil \bar  x_{ki,p}(l)\rceil=1$.  In addition, $\bar\beta_{kij}(l,m) \bar x_{ki,p}(l)\in (0,1]$  together with Constraints \eqref{eqn:mix-x-multi-cont} and \eqref{eqn:mix-beta-multi-cont}  also imply $\beta_{kij}(l,m)  x_{ki,p}(l)=\lceil \bar\beta_{kij}(l,m)\rceil \lceil \bar  x_{ki,p}(l)\rceil=1$. Thus, in this case, we can show $\{x_{ij,p}(l)\},\{ \beta_{kij}(l,m)\}$ satisfies Constraint \eqref{eqn:mix-x-inter-multi}. Note that under the continuous constraints $x_{ij,p}(l)\in [0,1]$ and $ \beta_{kij}(l,m)\in [0,1]$, the above two cases are the only two possible cases. Therefore, we can show  $\{z_{ij}\},\{z_{ij}(l)\}, \{f_{ij,p}^{t}(l)\}, \{ x_{ij,p}(l)\},\{\beta_{kij}(l,m)\}$ is a feasible solution to  Problem \ref{Prob:new-low-cont-multi}. 
 
 Therefore, we complete the proof of Lemma \ref{Lem:relation-cont-relax}.


\begin{thebibliography}{10}
\providecommand{\url}[1]{#1}
\csname url@samestyle\endcsname
\providecommand{\newblock}{\relax}
\providecommand{\bibinfo}[2]{#2}
\providecommand{\BIBentrySTDinterwordspacing}{\spaceskip=0pt\relax}
\providecommand{\BIBentryALTinterwordstretchfactor}{4}
\providecommand{\BIBentryALTinterwordspacing}{\spaceskip=\fontdimen2\font plus
\BIBentryALTinterwordstretchfactor\fontdimen3\font minus
  \fontdimen4\font\relax}
\providecommand{\BIBforeignlanguage}[2]{{%
\expandafter\ifx\csname l@#1\endcsname\relax
\typeout{** WARNING: IEEEtran.bst: No hyphenation pattern has been}%
\typeout{** loaded for the language `#1'. Using the pattern for}%
\typeout{** the default language instead.}%
\else
\language=\csname l@#1\endcsname
\fi
#2}}
\providecommand{\BIBdecl}{\relax}
\BIBdecl

\bibitem{KM03}
R.~Koetter and M.~M\'edard, ``Beyond routing: an algebraic approach to network
  coding,'' in \emph{INFOCOM 2002. Twenty-First Annual Joint Conference of the
  IEEE Computer and Communications Societies. Proceedings. IEEE}, vol.~1, Oct
  2002, pp. 122--130 vol.1.

\bibitem{LYC03}
S.-Y. Li, R.~Yeung, and N.~Cai, ``Linear network coding,'' \emph{Information
  Theory, IEEE Transactions on}, vol.~49, no.~2, pp. 371--381, Feb 2003.

\bibitem{Hoetal06}
T.~Ho, M.~M\'edard, R.~Koetter, D.~Karger, M.~Effros, J.~Shi, and B.~Leong, ``A
  random linear network coding approach to multicast,'' \emph{Information
  Theory, IEEE Transactions on}, vol.~52, no.~10, pp. 4413--4430, Oct 2006.

\bibitem{DFZ05}
R.~Dougherty, C.~Freiling, and K.~Zeger, ``Insufficiency of linear coding in
  network information flow,'' in \emph{Information Theory, 2005. ISIT 2005.
  Proceedings. International Symposium on}, Sept 2005, pp. 264--267.

\bibitem{twounicastKamath2014}
S.~Kamath, D.~N. Tse, and C.-C. Wang, ``Two-unicast is hard,'' in
  \emph{Information Theory, 2014. ISIT 2014. IEEE International Symposium on},
  June 2014.

\bibitem{Lunetal06}
D.~S. Lun, N.~Ratnakar, M.~M{\'e}dard, R.~Koetter, D.~R. Karger, T.~Ho,
  E.~Ahmed, and F.~Zhao, ``Minimum-cost multicast over coded packet networks,''
  \emph{Information Theory, IEEE Transactions on}, vol.~52, no.~6, pp.
  2608--2623, 2006.

\bibitem{Lun04networkcoding}
D.~S. Lun, M.~M\'edard, T.~Ho, and R.~Koetter, ``Network coding with a cost
  criterion,'' in \emph{in Proc. 2004 International Symposium on Information
  Theory and its Applications (ISITA 2004)}, 2004, pp. 1232--1237.
  
  \bibitem{CUI2015ISITreport}
\BIBentryALTinterwordspacing
Y.~Cui, M.~M\'edard, E.~Yeh, D.~Leith, and K.~Duffy. (2015, July) A linear
  network code construction for general integer connections based on the
  constraint satisfaction problem. ISIT Technical Report. [Online]. Available:
  \url{http://arxiv.org/abs/1502.06321}
\BIBentrySTDinterwordspacing

\bibitem{POEM10}
A.~ParandehGheibi, A.~Ozdaglar, M.~Effros, and M.~M\'edard, ``Optimal reverse
  carpooling over wireless networks - a distributed optimization approach,'' in
  \emph{Information Sciences and Systems (CISS), 2010 44th Annual Conference
  on}, March 2010, pp. 1--6.

\bibitem{ShroffBButterfly}
C.-C. Wang and N.~Shroff, ``Beyond the butterfly - a graph-theoretic
  characterization of the feasibility of network coding with two simple unicast
  sessions,'' in \emph{Information Theory, 2007. ISIT 2007. IEEE International
  Symposium on}, June 2007, pp. 121--125.



\bibitem{cfl}
K.~Duffy, C.~Bordenave, and D.~Leith, ``Decentralized constraint
  satisfaction,'' \emph{Networking, IEEE/ACM Transactions on}, vol.~21, no.~4,
  pp. 1298--1308, Aug 2013.

\bibitem{Bertsekasbooknonlinear:99}
D.~P. Bertsekas, \emph{Nonlinear Programming}, 2nd~ed.\hskip 1em plus 0.5em
  minus 0.4em\relax Belmont, MA: Athena Scientific, 1999.



\end{thebibliography}


\end{document}